%% file: hashgraph.tex
\documentclass{acmconf}
\usepackage{amssymb,url,tikz,nicefrac,code}

\title{Verifying the Hashgraph Consensus Algorithm}
\affiliation{Carnegie Mellon University}
\author{Karl Crary}

\input{defs}

\begin{document}

\maketitle{}

\begin{abstract}
The Hashgraph consensus algorithm is an algorithm for asynchronous Byzantine fault tolerance
intended for distributed shared ledgers.  Its main distinguishing
characteristic is it achieves consensus without exchanging any extra
messages; each participant's votes can be determined from public
information, so votes need not be transmitted.

In this paper, we discuss our experience formalizing the Hashgraph
algorithm and its
correctness proof using the Rocq proof assistant.  The paper is
self-contained; it includes a complete discussion of the algorithm and
its correctness argument in English.
\end{abstract}

\section{Introduction}

Byzantine fault-tolerance is the problem of coordinating a distributed
system while some participants may maliciously break the rules.  
We are particularly interested in the asynchronous version, in which
the weakest possible assumptions are made regarding the network.
The problem is at the center of a variety of new applications such as
cryptocurrencies.  Such applications rely on {\em distributed shared
ledgers,} a form of Byzantine fault-tolerance in which a set of
transactions are assigned a place in a globally-agreed total order
that is {\em immutable.} The latter means that once a transaction
enters the order, no new transaction can enter at an earlier position.

A distributed shared ledger makes it possible for all participants to
agree, at any point in the order, on the current owner of a digital
commodity such as a unit of cryptocurrency.  A transaction
transferring ownership is valid if the commodity's current owner
authorizes the transaction.  (The authorization mechanism---presumably
using a digital signature---is beyond the scope of the ledger itself.)
Because the order is total, one transaction out of any pair has
priority.  Thus we can show that a commodity's chain of ownership is
uniquely determined.  Finally, because the order is immutable, the
chain of ownership cannot change except by adding new transactions at
the end.

Algorithms for Byzantine consensus (under various assumptions) have
existed for some time, indeed longer than the problem has been
named~\cite{pease+:agreement-in-the-presence-of-faults,lamport+:byzantine-generals}.
Practical algorithms are more recent; in 1999, Castro and
Liskov~\cite{castro+:practical-byzantine-fault-tolerance} gave an
algorithm that when installed into the NFS file system slowed it only 3\%.
As Byzantine consensus algorithms have become more practical, they
have been tailored to specific applications.  Castro and Liskov's
algorithm was designed for fault-tolerant state machine
replication~\cite{schneider:state-machine-approach} and probably would not perform well under
the workload of a distributed shared ledger.

However, recently there have arisen asynchronous, Byzantine fault-tolerance
algorithms suitable for distributed shared ledgers, notably
HoneyBadgerBFT~\cite{miller+:honey-badger}, BEAT~\cite{duan+:beat},
and---the subject of this paper---Hashgraph~\cite{baird:hashgraph}.
The former two each claim to be the first practical
asynchronous BFT algorithm (with differing measures of
practicality).  Hashgraph does not claim to be first, but is also
practical.

In parallel with that line of work has been the development of
distributed shared ledgers based on {\em proof of work,} beginning
with Bitcoin~\cite{nakamoto:bitcoin}.  The idea behind proof of work
is to maintain agreement on the ledger by maintaining a list of blocks
of transactions, and to ensure that the list does not become a tree.
To ensure this, the rules state that (1) the longest branch defines
the list, and (2) to create a new block, one must
first solve a mathematical problem that takes the list's old head as one
of its inputs.  The problem's solution is much easier to verify than to obtain,
so when one learns of a new block, one's incentive is to restart work
from the new head rather than continue work from the old head.

Bitcoin and some of its cousins are widely used, so in a certain sense
they are indisputably practical.  They are also truly permissionless,
in a way that the BFT algorithms, including Hashgraph, cannot quite
claim.  Nevertheless, they offer severely limited throughput.  Bitcoin
is limited to seven transactions per second and has a latency of one
hour,\footnote{other proof-of-work systems do a little better} while
its BFT competitors all do several orders of magnitude
better.  Proof-of-work systems are also criticized for being wasteful:
an enormous amount of electricity is expended on block-creation
efforts that nearly always fail.  Finally---more to the point of this
paper---the theoretical properties of proof of work are not well
understood.

\medskip

The Hashgraph consensus algorithm is designed to support high-performance applications of a
distributed shared ledger.  Like the other BFT systems, it is several
orders of magnitude faster than proof of work.  Actual performance
depends very much on configuration choices ({\em e.g.,} how many
peers, geographic distribution, tradeoff between latency and
throughput, etc.), but in all configurations published in Miller, {\em
et. al}~\cite{miller+:honey-badger} (for HoneyBadgerBFT) and Duan, {\em et
al.}~\cite{duan+:beat} (for BEAT), the Hashgraph algorithm equals or exceeds the published
performance figures~\cite{baird+:hashgraph-measurements}.
A frequently cited throughput goal is to equal the Visa credit-card
network.  According to Visa's published figures, Hashgraph can handle
Visa's average load\footnote{3200 transactions per second in
2015~\cite{visa:average-tps-2015}.} and is in the ballpark of Visa's
claimed surge capacity.\footnote{65,000 transactions per second in
2017~\cite{visa:surge-tps-2017}.}

The key to the Hashgraph algorithm's performance is achieving nearly zero
communications overhead.  Previous BFT systems exchange messages to
achieve consensus, but Hashgraph does not.  In Hashgraph, no peer
sends any messages it did not intend to send anyway.  Moreover, the
overhead in each message it does send is light; each message consists mostly of
transactions that are new to the recipient.  In addition to the
transaction payload, a message contains an array of sequence numbers
(to keep track of which blocks each peer has seen), and the
information needed by consensus: just two hashes, a digital
signature, and a timestamp.

In the Hashgraph algorithm, peers achieve
consensus by voting, but each vote is fully determined by publicly
available information, so peers can determine each other's votes
without communicating with them.  This allows the consensus election
to be carried out {\em virtually,} with no extra messages exchanged.
There is no magic here; it still takes multiple cycles of
communication to achieve consensus, but every message is just an
ordinary block filled with new transactions.  Thus, consensus
is essentially invisible from a throughput perspective.

Another advantage of the Hashgraph algorithm is it requires no sophisticated
cryptography.  The only requirements are digital signatures and
cryptographic hashes, both of which are now commonplace and have
highly optimized implementations available.

\paragraph{This work}

The Hashgraph consensus algorithm has been realized as a large-scale,
open-access, commercial system called Hedera~\cite{baird+:hedera}, so
there is considerable interest in machine-verifying that it is correct.  In
this paper, we discuss the first steps toward doing so.  Using the Rocq
proof assistant~\cite{coq}, we formalized the batch-mode algorithm given in
Baird~\cite{baird:hashgraph} and developed a machine-checkable proof
of its correctness.  As usual when one formalizes a human-language
proof, we found a few errors, but they were minor and easily
corrected.  The algorithm implemented in Hedera is an online
algorithm, inspired by the batch-mode algorithm discussed here, but
obviously a bit different.\footnote{As it happens, none of the minor
  errors we found appear to affect the implemented online algorithm.}  We will
discuss some of the differences in
Section~\ref{sec:future-work}.

We begin by giving an informal overview of the algorithm, to build
intuition.  Then we give a human-language mathematical definition of the algorithm
and prove its properties, in Section~\ref{sec:algorithm}.
We discuss the formalization in Rocq starting in
Section~\ref{sec:formalization}.

\subsection{Hashgraphs in overview}

A {\em hashgraph\/} is a directed graph that summarizes who has said
what to whom.  Each peer maintains a hashgraph reflecting the
communications it is aware of.  In general, each peer knows a
different subset of the true graph, but because of digital signatures
and cryptographic hashes, they cannot disagree about the information
they have in common.

The nodes of a hashgraph are {\em events.}  Each event is created by a
particular peer.  Except for each peer's initial event, each event has
two parent events.  One parent has the same creator (we call that one the
{\em self-parent}), and the other has a different creator (we call it
the {\em other-parent}).  Honest peers do not create {\em forks},
where a fork is defined as two events with the same creator in which
neither is a self-ancestor of the other.\footnote{A self-ancestor is
an ancestor using only self-parent edges.}  In other words, the events
created by an honest peer will form a chain.

We can visualize a hashgraph as shown in Figure~\ref{fig:graph}.  In
this example, all peers are behaving honestly.

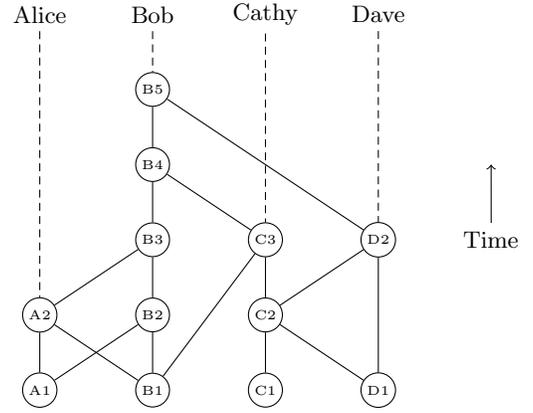
\begin{figure}
\hfill
\begin{tikzpicture}[xscale=1.5]
\newcommand{\mkevent}[1]{node[circle,inner sep=1.5pt,draw] (#1) {\tiny #1}}
\newcommand{\mkpeer}[2]{(#2,5) node (#1) {#1}}
\path(0,0) \mkevent{A1}
++(0,1) \mkevent{A2}
\mkpeer{Alice}{0};
\path(1,0) \mkevent{B1}
++(0,1) \mkevent{B2}
++(0,1) \mkevent{B3}
++(0,1) \mkevent{B4}
++(0,1) \mkevent{B5}
\mkpeer{Bob}{1};
\path(2,0) \mkevent{C1}
++(0,1) \mkevent{C2}
++(0,1) \mkevent{C3}
\mkpeer{Cathy}{2};
\path(3,0) \mkevent{D1}
++(0,2) \mkevent{D2}
\mkpeer{Dave}{3};
\draw (A1) -- (A2);
\draw[densely dashed] (Alice) -- (A2);
\draw (B1) -- (B2) -- (B3) -- (B4) -- (B5);
\draw[densely dashed] (Bob) -- (B5);
\draw (C1) -- (C2) -- (C3);
\draw[densely dashed] (Cathy) -- (C3);
\draw (D1) -- (D2);
\draw[densely dashed] (Dave) -- (D2);
\draw (A1) -- (B2);
\draw (B1) -- (A2) -- (B3);
\draw (B1) -- (C3) -- (B4);
\draw (D1) -- (C2) -- (D2) -- (B5);
\path (4,2) node (Time) {Time};
\draw[->] (Time) -- (4,3);
\end{tikzpicture}
\caption{A Hashgraph}
\label{fig:graph}
\figfoot
\end{figure}

In the Hashgraph network, each peer periodically chooses another
at random and sends that peer its latest event.  The recipient then
creates a new event, with its own latest event as self-parent and the
event it just received as other-parent.  Every event is digitally
signed, so there can be no dispute over who created it.  Each event
also contains the hashes of its two parents, so there is no dispute
over parentage either.  The recipient of the event will request from
the sender any of the event's ancestors that it does not have.  (For
simplicity, we will ignore that part of the protocol in what follows.)
The recipient also puts a timestamp into its new event, which is
ultimately used to determine a consensus timestamp for other events.

Finally, each event contains a payload of transactions.  When a peer
wishes to inject a new transaction into the network, it stores it in a
buffer of outgoing transactions.  The next time it creates an event
(this happens multiple times per second), it uses the contents of its
buffer as the new event's payload.  Transactions are just along for
the ride in the consensus algorithm, so we will discuss them little.

In the example, Dave sent D1 to Cathy, resulting in Cathy creating
C2.  Then Cathy sent C2 back to Dave, resulting in Dave creating D2.
Bob sent B1 to both Alice and Cathy, resulting in A2 and C3.  At about
the same time, Alice sent A1 to Bob, resulting in B2.  Alice sent A2
to Bob, resulting in B3.  Cathy sent C3 to Bob, resulting in B4.
Finally, Dave sent D2 to Bob, resulting in B5.

The algorithm partitions events into rounds, in a manner that is easy
for all peers to agree on.  The first event created by a peer in a
round is called a {\em witness.}  (Note that dishonest peers may have
multiple witnesses in a single round.)  A witness that is quickly
propagated to most peers is called {\em famous.}  Identifying the
famous witnesses is the main job of the consensus algorithm.

Each round, the algorithm selects one famous witness from each peer
that has one to be a {\em unique famous witness.}
An honest peer will have at most one witness per round, so if that
witness turns out to be famous, it will also be unique.  If a
dishonest peer happens to have multiple famous witnesses, one of them
is selected.

We say that an event has been {\em received by the network\/} in the
first round in which it is an ancestor of all the unique famous
witnesses.  The round an event is received is the primary determiner of an event's
place in the order.  Ties are broken using a consensus timestamp that
is computed using the unique famous witnesses.  Any remaining ties are
broken in an arbitrary but deterministic way.  Finally, the ordering
of transactions is determined by the ordering of the events in which
they reside, with transactions in the same event ordered by their
position in the payload.

\section{The Algorithm}
\label{sec:algorithm}

We will begin by reviewing the Hashgraph consensus
algorithm~\cite{baird:hashgraph}, and defer discussion of the
formalization until Section~\ref{sec:formalization}.  We will develop
the algorithm in pieces, establishing the properties of those pieces
as we go.  But first, the algorithm relies on the following
assumptions:

\begin{enumerate}
\item
The network is asynchronous.  However, the adversary cannot disconnect
honest peers indefinitely; every honest peer will eventually communicate
with every other honest peer.

\item
Every peer can determine any event's creator and parents.  In the real
world, this assumption means that the adversary has insufficient
computing power to forge a digital signature, or to find two events
with the same hash.

\item
A supermajority (defined to mean more than two-thirds) of
the peers are honest.

\item
The coin mechanism discussed below satisfies a probability
assumption.  (Section~\ref{sec:termination}.)

\item
Various uncontroversial mathematical assumptions hold.  These are
discussed in Section~\ref{sec:formalization}.  For example, the
parent relation is well-founded.  These assumptions would ordinarily
go unremarked, but formal verification requires that they be
specified.
\end{enumerate}

For convenience, we will also assume that there are at least two
peers.  (If there is only one peer, the consensus problem is trivial.)

Suppose $x$ and $y$ are events.  We will write $x \anc y$ when $x$ is
a (non-strict) ancestor of $y$, and $x \sanc y$ when $x$ is a strict
ancestor of $y$.  We say $x$ is a self-ancestor of $y$ (and write $x
\selfanc y$) when $x$ is an ancestor of $y$ using only self-parent
edges, and we write $x \selfsanc y$ for the strict version.  In all
cases, note that the older event is on the left.

We will refer to events created by an honest peer as {\em honest
  events.}  We will refer to a set of events whose creators constitute
a supermajority of the peers as a {\em supermajor set.}  We will take
$N$ to be the total number of peers.

\subsection{Seeing}

Recall that events $x$ and $y$ form a {\em fork\/} if they have the
same creator, and if $x \not\selfanc y$ and $y \not\selfanc x$.  Note
that forks are defined so that no event forms a fork with itself.  The
main property of honest peers is they never create forks.

We say that $y$ {\em sees\/} $x$ (written $x \see y$) if (1) $x \anc
y$ and (2) there does not exist any fork $z, z'$ such that $z, z' \anc
y$ and $\creator(x) = \creator(z)$.  Sees is the same as ancestor except that 
when an event observes a fork, it blacklists the fork's creator and will see none of
the fork's creator's events.

This brings us to the algorithm's main technical concept:

\begin{defn}
We say that $y$ {\em strongly sees\/} $x$ (written $x \stsee y$) if there
exists a supermajor set $Z$, such that for all $z \in Z$, $x \see z
\anc y$.
\end{defn}

Informally, for $y$ to strongly see $x$ means that $x$ has made its
way to most of the network, within the subgraph visible to $y$, with
hardly anyone observing $x$'s creator cheat.

For example, consider Figure~\ref{fig:graph}.  Every peer is behaving
honestly, so every ancestor is seen.  Then we can say that B4
strongly sees B1, using A2, B4, and C3 as intermediaries.  (B4 has no
intermediary on Dave, but it does not need one.  Three intermediaries
is enough, since $3 > \twothirds \cdot 4$.)
B4 also strongly sees D1, using B4, C3, and D1
as intermediaries.  B4 does not strongly see A1 or C1, as it has only
two intermediaries for each.  However, B5 does strongly see C1, using
intermediaries B5, C2, and D2.  In most of these cases, other choices
could be made for the intermediaries as well.

We can state some useful properties of strongly-seeing:

\begin{lemma}
\label{lem:stsee-properties}
\zilch
\begin{enumerate}
\item
If $x \stsee y$ then $x \sanc y$.

\item
If $x \selfanc y \stsee z$ then $x \stsee z$.

\item
If $x \stsee y \anc z$ then $x \stsee z$.
\end{enumerate}
\end{lemma}
\begin{proof}
For (1), suppose $x \stsee y$.  We can immediately see that $x \anc
y$, but we also have $x \neq y$.  Since $y$ strongly sees $x$, there must be
intermediaries on a supermajority of peers, at least one of which must
not be $y$'s creator.  (Since we assume there are at least two peers,
a supermajority is at least two.)  Let $z$ be one such.  Then $x \see
z \sanc y$, so $x \sanc y$.

For (2), observe that $x \selfanc y \see z$ implies $x \see z$, since
$x$ and $y$ have the same creator, so any fork on $x$'s creator is
a fork on $y$'s creator.  (3) is immediate by transitivity of ancestor.
\end{proof}

We can now state the main technical lemma, which states that at most
one side of a fork can be strongly seen, even by different events:

\begin{lemma}[Strongly Seeing]
\label{lem:strongly-seeing}
Suppose $x \stsee v$ and $y \stsee w$.  Then $x$ and $y$ do not form a
fork.
\end{lemma}
\begin{proof}
Let $Z = \{ z \sbar x \see z \anc v \}$ and $Z' = \{ z' \sbar y \see
z' \anc w \}$.  By the definition of strongly-sees, $Z$ and $Z'$ are
supermajor sets.  Also recall
that a supermajority of peers are honest.  Any three supermajorities
have an element in common.  Therefore there exist events $z \in Z$ and
$z' \in Z'$ such that $z$ and $z'$ share the same creator, and that creator
is honest.  Since their creator is honest, either $z \selfanc z'$ or
$z' \selfanc z$.  Assume the latter.  (The former case is similar.)
Then $x, y \anc z$.  Thus $x$ and $y$ cannot form a fork, since $x
\see z$.
\end{proof}

The definition of strongly-sees limits the influence of dishonest
peers, since at most one side of a fork can be strongly seen.  This is
helpful in various ways; one important way is it prevents dishonest
peers from obtaining extra votes in an election by creating extra
witnesses, since votes must be strongly seen to count.

\subsection{Rounds}
\label{sec:rounds}

Every event is assigned to a {\em round.}  (Note that the round an
event belongs to is different from the round that it is received by
the network.)  Initial events belong to round zero.  For a
non-initial event $x$, let $i$ be the maximum of $x$'s parents'
rounds.  Then $x$ belongs to round $i$, unless $x$ strongly sees
events in round $i$ on a supermajority of peers, in which case $x$
belongs to round $i + 1$.

For example, consider Figure~\ref{fig:graph}.  A1, B1, C1, and D1 are
initial events, so they belong to round 0.  For every other event, the
maximum of the parents' rounds is 0.  Thus, every other event also
belongs to round 0, except B5.  As noted above, B4 can strongly see
B1 and D1, but not A1 or C1, so it remains in round 0.  However, B5
strongly sees B1, C1, and D1, so it advances to round 1.

The first event in a round on each peer (or first events, in the case of a
dishonest peer) are called {\em witnesses.}  Witnesses are the events
that can cast votes in an election.

In Figure~\ref{fig:graph}, the witnesses are A1, B1, C1, and D1
(round~0), and B5 (round 1).

Note that if a peer has not heard from anyone in a while, it is
possible for it to skip one or more rounds when it catches up.  Thus,
a peer might not have a witness in any particular round.  Conversely,
a dishonest peer can have multiple witnesses in a round, but (by
Lemma~\ref{lem:strongly-seeing}) at most one of them can be strongly
seen.

\smallskip

Two important properties of rounds follow more-or-less directly:

\begin{lemma}
\label{lem:round-elim}
Suppose $i < j$ and $x$ belongs to round $j$.  Then there exists a supermajor
set of round $i$ witnesses $W$, such that for all $w \in W$, $w \stsee x$.
\end{lemma}

\begin{lemma}
\label{lem:later-round}
Suppose $x$ and $y$ are events.  If there exists a supermajor set $W$
such that for all $w \in W$, $x \anc w \stsee y$, then $y$'s round is
strictly later than $x$'s round.
\end{lemma}

The main property we wish to establish about rounds is {\em progress,}
which says that every round is inhabited.

\begin{lemma}[Broadcast]
\label{lem:broadcast}
Suppose $x$ is an honest event.  Then there exists an honest event $y$
such that, in its ancestry, $x$ has reached every honest peer.  That
is, for every honest peer $a$, there exists $z$ created by $a$, such
that $x \anc z \anc y$.
\end{lemma}
\begin{proofsketch}
By induction on the number of honest peers, using the assumption that
every honest peer eventually communicates with every other
honest peer.
\end{proofsketch}

\begin{lemma}[Progress]
\label{lem:progress}
Every round is inhabited.
\end{lemma}
\begin{proof}
The proof is by induction on $i$.  The case $i=0$ is immediate, so
suppose $i > 0$.  By induction there exists an inhabitant $x$ of round
$i-1$.  Using Lemma~\ref{lem:broadcast}, there exists $y$ such that
$x$ has reached every honest peer in $y$'s ancestry.  And again, there
exists $z$ such that $y$ has reached every honest peer in $z$'s
ancestry.

Suppose $a$ is an honest peer.  We show that there exists $v$
created by $a$ that $x \anc v \stsee z$.  Since the honest peers
constitute a supermajority, it follows by Lemma~\ref{lem:later-round}
that $z$'s round is later than $i-1$, and is thus at least $i$.
If $z$'s round is later than $i$, it is easy to show that it has an
ancestor in round $i$.

By the specification of $y$, there exists $v$ created by $a$ such that $x
\anc v \anc y$.  It remains to show $v \stsee z$.  For every honest
peer $b$, by the specification of $z$ there exists $w$ created by $b$
such that $y \anc w \anc z$.  Since $a$ is honest, $v \see w \anc z$.  Since
the honest peers constitute a supermajority, $v \stsee z$.
\end{proof}

\subsection{Voting}
\label{sec:voting}

One cannot use the raw witnesses to determine the round an event was
received by the network, because one can never be sure that one has
observed all the witnesses.  Instead, we will use the notion of a {\em
  famous\/} witness---eventually everyone will know all the famous
witnesses in any given round.  The main function of the consensus
algorithm is to determine which witnesses are famous.

Fame is determined by an election.  Votes are cast by witness events,
not by peers.  This is important because one peer has no way of
knowing what another {\em peer\/} knows.  In contrast, if a peer is
aware of an {\em event\/} at all, it knows that event's entire ancestry.
Thus, if voting is deterministic, each witness's vote can be
determined by everyone who is aware of the witness, without any
additional communication.  However, votes will be ``collected'' only from
strongly-seen witnesses, so there can be at most one meaningful voter per
peer per round.

We will refer to the witness whose fame is being determined as the {\em
  candidate,} and each witness that is casting votes on the
candidate's fame as the {\em voter.}  The election begins $d$ rounds
after the candidate's round, where $d$ is a parameter that is at least
1.  Thus if the candidate belongs to round $i$, the voters in the
first round of the election belong to round $i+d$.

In the first round of an election, each voter will vote yes if the
candidate is among its ancestors.  (In essence, you think someone is
famous if you have heard of them.)  In successive rounds, each voter
votes the way it observed the majority vote in the previous round.

We say that a round of an election is {\em nearly unanimous\/} if
voters on a supermajority of peers vote the same way.  (Note that this
is a stronger condition than merely a supermajority of the voters,
since some peers might not be voting.)
If a voter ever observes a nearly unanimous result in one
round, then the vote in the next round will be unanimous, since any
two supermajorities must have a majority of their elements in common.  Clearly,
once the election becomes
unanimous, it will stay so.  Thus we can end the election as soon as
any event observes a nearly unanimous result in the previous round.

\paragraph{Coins}

Under normal circumstances, this process will come to consensus
quickly, but an adversary with sufficient control of the network can
prevent it.  Deterministically, the problem is
insurmountable~\cite{fischer+:flp}, but, as usual, it can be solved
with randomization.

Every $c$ rounds (a parameter at least $d + 3$), the election will
employ a coin round: Any voter who sees a nearly unanimous result the
previous round will continue to vote with the majority.  However, the
remaining voters will determine their votes using a coin flip.  All
the voters who saw a nearly unanimous result will certainly vote the
same way, and eventually---by chance---the honest coin flippers will also
vote that way.  Two rounds later, every voter will see a unanimous
result and the election will end.\footnote{This is an important
  theoretical property, but as a practical matter, in the unlikely
  event an adversary has enough control over the network to force coin
  rounds, it will be able to grind the algorithm to a halt for rounds
  exponential in the number of peers.}

\begin{definition}[Voting]
Suppose $x$ is a round $i$ witness and $y$ is a round $j$ witness,
with $i + d \leq j$.  Then $\vote(x, y)$ (that is, $y$'s vote on $x$'s
fame) and $\election(x, y, t, f)$ (that is, the votes on $x$'s fame
observed by $y$ are $t$ yeas and $f$ nays) are defined simultaneously as
follows:
\begin{itemize}
\item
If $i + d = j$ then $\vote(x,y)$ is $\yes$ if $x \anc y$, and $\no$
otherwise.

\item
If $i + d \lt j$ and $(j - i) \modulo c \neq 0$ and $\election(x, y, t, f)$, then $\vote(x,y)$ is
$\yes$ if $t \geq f$ and $\no$ otherwise.

\item
If $i + d \lt j$ and $(j - i) \modulo c = 0$ and $\election(x, y, t, f)$,
then
\[
\vote(x, y) =
\left\{
\begin{array}{ll@{}l}
\yes & t & {} \gt \twothirds \cdot N
\\
\no & f & {} \gt \twothirds \cdot N
\\
\coin(y) & \multicolumn{2}{l}{\mbox{otherwise}}
\end{array}
\right.
\]

\item
If $i + d \lt j$, then $\election(x, y, t, f)$ holds if and only if:
\[
\begin{array}{lll}
W & = & \{ \mbox{round $j-1$ witnesses $w$ such that $w \stsee y$} \}
\\
t & = & |\{ w \in W \sbar \vote(x, w) = \yes \}|
\\
f & = & |\{ w \in W \sbar \vote(x, w) = \no \}|
\end{array}
\]
\end{itemize}
\end{definition}

In the above definition, $\coin(y)$ is a pseudo-random coin flip
computed by drawing a bit from the middle of a cryptographic hash of
$y$.  It is important that the coin flip is pseudo-random, not truly
random, so that other peers can reproduce it.

\begin{definition}[Decision]
Suppose $x$ is a round $i$ witness and $y$ is a round $j$ witness,
where $i + d \lt j$ and $j$ is not a coin round (that is, $j - i \modulo c
\neq 0$).  Suppose further that $\election(x, y, t, f)$.  Then $y$
decides $\beta$ on $x$'s fame (written $\decide(x, y, \beta)$), if $t
\gt \twothirds \cdot N$ and $\beta = \yes$, or if $f \gt \twothirds
\cdot N$ and $\beta = \no$.
\end{definition}

The outcome of the election is determined as soon as any peer decides,
but other peers might not realize it right away.

\subsection{Consensus}

For the algorithm to work, we require four properties:

\begin{enumerate}
\item
Decisions are pervasive: once one peer decides someone's fame, that
decision propagates to every other peer.
(Corollary~\ref{cor:propagation}.)

\item
Every round will have at least one famous witness.
(Theorem~\ref{thm:famous-exists}.)

\item
Late arrivals are not famous: if a witness is not well disseminated
within $d+1$ rounds (the earliest a decision can be made), it will not
be famous.
(Corollary~\ref{cor:no-late-fame}.)

\item
Termination: eventually every witness will have its fame decided.
(Theorem~\ref{thm:termination}.)
\end{enumerate}

\begin{lemma}[Decision-Vote Consistency]
\label{lem:decision-vote-consistency}
Suppose $w$, $x$, and $y$ are witnesses, where $x$ and $y$ belong to
the same round.  If $\decide(w, x, \beta)$ and $\vote(w, y) = \beta'$
then $\beta = \beta'$.
\end{lemma}
\begin{proof}
Let $j$ be the round of $x$ and $y$.
Observe that $\vote(w, y)$ is given by the second case of the
definition of voting.  Unpacking the definitions, let $\election(w, x,
t, f)$ and $\election(w, y, t', f')$.  Let $V_z = 
\{ \mbox{round $j-1$ witnesses $v$ such that $v \stsee z$} \}$.
Recall that at most one witness per peer per round can be strongly seen.

Suppose $\beta = \yes$.  (The other case is similar.)  Then $|T| >
\twothirds \cdot N$ where $T = \{ v \in V_x \sbar \vote(w, v) = \yes
\}$.  Thus $T$ is a supermajor set.  By
Lemma~\ref{lem:round-elim}, $V_y$ is also a supermajor set.  Observe
that $T$ and $V_y$ are both sets of round $j-1$ witnesses.  Any two
supermajorities of the same set have a majority in common.  Thus a
majority of $V_y$ is in common with $T$, but $T$ all vote yes.  Hence $t' >
f'$ so $\beta' = \yes$.
\end{proof}

Note that the proof illustrates why decisions cannot be made in coin
rounds.  In a coin round, $y$ would have to see a supermajority to
determine its vote, not merely a majority, and we cannot guarantee a
supermajority.  Thus the lemma would fail to hold.

\begin{corollary}[Propagation]
\label{cor:propagation}
Suppose $x$ makes a decision on $w$.  Then in every round after
$x$'s, except coin rounds, every witness decides the same way as $x$ did.
\end{corollary}
\begin{proofsketch}
Let $j$ be the round of $x$.  By
Lemma~\ref{lem:decision-vote-consistency}, every witness in round $j$
votes $\beta$.  It is easy to see that unanimity persists.  Thus every
witness in any round after $j$ decides $\beta$, unless it is
prohibited from doing so because it is a coin round.
\end{proofsketch}

\noindent
An easy corollary is that all decisions agree:

\begin{corollary}[Consistency]
\label{cor:consistency}
Suppose $\decide(w, x, \beta)$ and $\decide(w, y, \beta')$.  Then
$\beta = \beta'$.
\end{corollary}
\begin{proof}
Let $z$ be a witness in a later round than $x$ and $y$ that is not a
coin round.  (Such a witness exists by Lemma~\ref{lem:progress}.)
By Corollary~\ref{cor:propagation}, $\decide(w, z, \beta)$ and
$\decide(w, z, \beta')$.  Hence $\beta = \beta'$.
\end{proof}

\begin{theorem}[Existence]
\label{thm:famous-exists}
Every round has at least one famous witness.
\end{theorem}
\begin{proof}
If $S$ is a set of events,
define $\backward(j, S) = \{ w \sbar \mbox{$w$ is a round $j$ witness
  and\ } \exists v \in S .\, w \stsee v \}$.

Suppose $S$ is an inhabited set of round $k$ witnesses and $j < k$.
By Lemma~\ref{lem:round-elim}, $\twothirds \cdot N < |\backward(j, S)| \leq N$.

Let $i$ be arbitrary, and let $z$ be an arbitrary witness in round
$i+d+2$, then let:
\[
\begin{array}{lll}
S_3 &=& \{ z \} \\
S_2 &=& \backward(i+d+1, S_3) \\
S_1 &=& \backward(i+d, S_2) \\
S_0 &=& \backward(i, S_1)
\end{array}
\]
Now consider a bipartite graph between $S_0$ and $S_1$, where there is
an edge between $x \in S_0$ and $y \in S_1$ if $x \anc y$.  Since
strongly-seeing implies ancestor, Lemma~\ref{lem:round-elim} tells us
that every event in $S_1$ is the terminus of more than $\twothirds
\cdot N$ edges.  Thus there are over $\twothirds \cdot N \cdot |S_1|$
edges total.  Since there are at most $N$ events in $S_0$, by the
pigeonhole principle there is at least one event in $S_0$ that is the
origin of over $\nicefrac{2}{3} \cdot |S_1|$ edges.  Let $x$ be such
an event.  We will show that $x$ is famous.

By the specification, $x$ is an ancestor of over two-thirds of the
events in $S_1$.  Thus, over two-thirds of $S_1$ will vote yes.  We
claim that every event in $S_2$ will vote yes.  It follows that $z$
will decide yes.  (We know $i+d+2$ will not be a coin round since $c
\geq d+3$ by assumption.)

Suppose $y \in S_2$.  The events sending votes to $y$ are exactly
$\backward(i+d, \{ y \})$ and, as above, $\twothirds \cdot N <
|\backward(i+d, \{ y \})|$.  Thus, the events in $S_1$ voting yes and
the events in $S_1$ sending votes to $y$ are both supermajor,
so they have a majority in common.  Thus $y$ will see more yeas than
nays, and will vote yes itself.
\end{proof}

\begin{lemma}
\label{lem:yes-vote-impl-anc}
If $\vote(x, y) = \yes$ then $x \anc y$.
\end{lemma}
\begin{proofsketch}
By well-founded induction on $y$ using the strict-ancestor order
($\sanc$).  In the first voting round, a voter will vote yes only if
it is a descendant of $x$.  Thereafter, a voter cannot vote yes
without receiving some yes votes (half in regular rounds and a third
in coin rounds), and one only receives votes from ancestors.
\end{proofsketch}

\begin{corollary}[No Late Fame]
\label{cor:no-late-fame}
Suppose $x$ is a round $i$ witness and $y$ is a round $j$ witness,
where $i + d < j$ and $(j - i) \modulo c \neq 0$.  If $x \not\anc y$ then
$\decide(x, y, \no)$.
\end{corollary}
\begin{proof}
Any voter sending votes to $y$ must be an ancestor of $y$, and
therefore it cannot be a descendant of $x$.  Thus, by
Lemma~\ref{lem:yes-vote-impl-anc}, $y$ will see only no votes.  Since
$j$ is a round in which one is permitted to decide, $y$ will decide no.
\end{proof}

Thus, if by round $i+d+1$ there remain any witnesses that have not
heard of $x$, then $x$ is not famous.  Consequently, a late arrival to
the network cannot be famous.

\subsubsection{Termination}
\label{sec:termination}

The heart of the termination argument is the observation that
eventually (with probability one) there will come a round in which all
of the honest coin flippers will agree with the previous round's
majority.  Making this observation rigorous involves some subtlety.
The first subtlety is the probability assumption itself:

First, we assume that honest coins come up heads with probability
$1/2$.  This is actually more subtle than it may sound, because coin
flips are actually only pseudo-random.  Nevertheless, we assume that
on honest peers they behave as though they are truly random.
Certainly any peer can easily dictate its event's coin flip (which is
drawn from a cryptographic hash of the event) by tweaking the event's
payload, but we assume honest peers will not do that.  Beyond that, in
principle an adversary might orchestrate the network in order to
control the payloads and thereby dictate the coin flips.  We assume
that it has insufficient power over the network to do so.

Second, we assume that honest coins are independent of probability
events that are already settled.  Suppose $P$ is a probability event
and $x$ is an honest Hashgraph event.  We say that $P$ is {\em settled
  before\/} $x$ if 
by the time $x$ was created, it was already determined whether
$P$ takes place or not.
If $P$ is settled before $x$, we assume that
$x$'s coin is independent of $P$
(since $P$ cannot look into the future and be affected by $x$, whereas $x$,
being a fair coin, does not look into the past.)

Moreover, if some $Q$ implies that
$P$ is settled before $x$, then $x$'s coin is {\em conditionally\/}
independent of $P$, assuming $Q$.

A simple consequence of the latter assumption is that honest coins are
independent, but it also follows that honest coins are independent of
the operating of the algorithm before the coin is flipped.

It is convenient to gather all our probabilistic reasoning into one lemma:

\begin{definition}
We say that {\em coins are good} if for every round-$i$ witness $x$,
there exists some round $j > i$ and some boolean $\beta$ such that:

\begin{enumerate}
\item
$(j - i) \modulo c = 0$ ({\em i.e.,} $j$ is a coin round),

\item
$\beta = {\sf yes}$ if and only if $y$ receives as many yes votes for
$x$ as no votes, where $y$ is the first round-$j$ witness, and

\item
for every honest round-$j$ witness $w$, if $w$ is timely in the sense
that it was created before any witnesses in round $(j+2)$, then
$\coin(w) = \beta$.
\end{enumerate}
\end{definition}

\begin{theorem}
With probability one, coins are good.
\end{theorem}
\begin{proofsketch}
For any round $j$, let $P_j$ be the probability event ``round $j$'s
first witness receives as many yes votes as no votes.''  
Let $Q_j$ be the probability event ``all honest, timely round-$j$ coin
flips agree with $P_j$.''

Observe that $P_j$ depends only on events created earlier than the
first round-$j$ witness, so it is settled before any round-$j$ event
is created.  Thus, by the probability assumption, $P_j$ is independent
of all honest round-$j$ coin flips.  Similarly, $\neg P_j$ is also
independent of all honest round-$j$ coin flips.  Therefore, for every
round $j$, the probability of $Q_j$ is at least $2^{-N}$, since there
can be no more than $N$ honest, timely witnesses in any round.  Each
of the $Q$s are also independent (the timeliness proviso
ensures that each coin round is settled before the next).
Therefore, $Q$ will eventually hold.
\end{proofsketch}

In any remaining results that rely on probability, we take good coins
as an antecedent.  This will allow us to use purely logical reasoning
in the sequel.

\begin{theorem}[Termination]
\label{thm:termination}
Suppose good coins.
Then for every witness $x$, there exists $z$ and $\beta$ such that
$\decide(x, z, \beta)$.
\end{theorem}
\begin{proof}
Let $i$ be the round of $x$.  Then let $j$, $y$, and $\beta$ be as
given by good coins.  Let $z$ be the first round-$(j+2)$ witness.

Suppose $w$ is an honest round-$j$ witness such that $w \leq z$.
We claim that $w$ will vote $\beta$.

\begin{itemize}
\item
Suppose $w$ sees a supermajority.  As we have seen before, when there
is a supermajority, every witness (such as $y$) will receive a
majority that agrees with the supermajority.  Thus, the supermajority
$w$ receives must agree with $\beta$.

\item
Alternatively, suppose $w$ uses its coin.  We must show that $w$ is
timely; it then follows that $\coin(w) = \beta$.  Since $w \leq z$,
$w$ was created earlier than $z$.  But $z$ is the first of its round,
so $z$ was created no later than any other round-$(j+2)$ witness.
\end{itemize}

Now suppose $u$ is a round-$(j+1)$ witness such that $u \leq z$.  The
electors sending votes to $u$ are a supermajor set, and a
supermajority of the peers are honest, so a majority of the votes come
from honest events.  All of them (by transitivity) are ancestors of
$z$, so (by the above) $u$ sees a $\beta$ majority.  Thus $u$ votes
$\beta$.  Consequently, $z$ receives only $\beta$ votes, so $z$
decides $\beta$.
\end{proof}

At this point we know that all peers can agree on the identity of the
famous witnesses: For any given witness, eventually someone will
decide (Theorem~\ref{thm:termination}).  Once someone decides, everyone
else will make the same decision in short order
(Corollary~\ref{cor:propagation}).  Once a peer has settled the fame
of every witness it has heard of, it can consider itself done, since
any additional witnesses it hasn't heard of are guaranteed not to be
famous (Corollary~\ref{cor:no-late-fame}).  Moreover, at least one
famous witness will exist (Theorem~\ref{thm:famous-exists}).

\subsection{Round Received}

At this point the real work is done.
Next, we identify at most one famous witness per peer as a {\em unique
  famous witness.}  If a peer has only one famous witness (as will be
always be the case for honest peers), that famous witness will be
unique.  In the unlikely event that a peer has multiple famous
witnesses, one of them is chosen to be unique in an arbitrary but
deterministic manner.  (Textual comparison of the data that expresses
the event will do fine.)

Then, using the unique famous witnesses, we define an event's {\em round
  received:}

\begin{definition}
Suppose $x$ is an event.  The {\em round that $x$ is received by the
  network\/} is the earliest round $i$ for which all the round-$i$
unique famous witnesses are descendants of $x$.
\end{definition}

\subsection{The Consensus Timestamp}

Suppose $x$ is an event that is received in round $i$.  We compute the
consensus timestamp for $x$ as the median of the timestamps assigned
to $x$ by the unique famous witnesses of round $i$.  (If there are an
even number of unique famous witnesses, we take the median to be the
smaller of the two central elements.  This provides the useful property
that the consensus timestamp is the timestamp from some particular
peer.)

Suppose $y$ is a round-$i$ unique famous witness.  The timestamp $y$
assigns to $x$ is the timestamp of the earliest $z$ such that $x \anc
z \selfanc y$.  Note that $z$ certainly exists, since $x \anc y$, and
that $z$ is uniquely defined since the self-ancestors of $y$ are
totally ordered.

\subsection{The Consensus Order}

We can now define the consensus order:

\begin{definition}
Suppose $x$ and $y$ are events.  Let $i$ and $i'$ be the rounds $x$ and
$y$ are received by the network.  Let $t$ and $t'$ be the consensus
timestamps of $x$ and $y$.  Then {\em $x$ precedes $y$ in the
consensus order\/} if:
\begin{itemize}
\item
$i < i'$, or

\item
$i = i'$ and $t < t'$, or

\item
$i = i'$ and $t = t'$ and $x$ is less than $y$ using the
  arbitrary comparison used to select unique famous witnesses.
\end{itemize}
\end{definition}

It is clear that the consensus order is a total order.  Since all
peers agree on the unique famous witnesses, they all agree on
the consensus order.  Finally, the order is immutable:  the unique
famous witnesses do not change once determined, and their ancestries
never change either.

\section{Formalization}
\label{sec:formalization}

Many interesting issues arose in the course of formalizing the
Hashgraph consensus algorithm.
Our work builds on lessons learned from a previous effort to verify
Hashgraph by Gregory Malecha and Ryan Wisnesky.  In that earlier
effort the algorithm was expressed by implementing it as a function in
Gallina (the Rocq specification language).  This allowed the code to
be directly executed within Rocq, and it was thought that such code
could be more easily related to the actual code in the commercial
Hashgraph implementation, Hedera.

However, the specification as a Gallina implementation was challenging to
work with.  Although the effort did not hit a show-stopping problem,
the concreteness of the implementation made it clumsy to work with.
Moreover, expressing the algorithm using recursive functions,
rather than inductive relations, meant that one was denied nice
induction principles.

In this work we started over with new definitions, making events and
other objects of interest abstract, and using inductive
relations for most definitions.  This allowed for streamlined proofs
that usually closely resemble the human-language proofs, and nice
induction principles in most cases.  (An important exception,
induction over votes, is discussed in Section~\ref{sec:formal-voting}.)

The axioms that define hashgraphs, real arithmetic, and probability
are just 305 lines.  This is the
trusted part of the development, apart from Rocq itself.

The full proof is 24 thousand lines of Rocq (version 8.9.1), including
comments and whitespace.  Not all of these lines are
allocated to matters that are conceptually interesting: About 3
thousand lines are dedicated to defining and establishing
properties of the cardinalities of sets, and of medians.  About 6
thousand more are dedicated to basic probability.

\medskip

We will discuss changes we made from the original algorithm in
Baird~\cite{baird:hashgraph}, and then touch on interesting points
that arose during the formalization.  The survey here will not be
complete; the complete development is available at:

\begin{center}
\url{cs.cmu.edu/~crary/papers/2026/hashgraph-formal.tgz}
\end{center}

\subsection{Changes from the original algorithm}

As typically happens when formalizing a human-language proof, we did
uncover a few errors, but they were minor and easily
corrected.\footnote{These errors pertain to the batch-mode algorithm
  discussed here, not to the online algorithm that the Hedera
  implementation is based on.  We discuss some of the differences in
  Section~\ref{sec:future-work}.  For example, the online algorithm
  defines strongly seeing somewhat differently.}  These
corrections are already reflected in Section~\ref{sec:algorithm}:

\begin{enumerate}
\item
The original definition of strongly-sees was a bit different: it said
that $y$ strongly sees $x$ if $x \see z \see y$ (for all $z$ in
some supermajor set).  But that definition didn't provide the
``stickiness'' property of Lemma~\ref{lem:stsee-properties} (that is,
$x \stsee y \anc w$ implies $x \stsee w$) since $w$ might observe a
fork that $y$ does not.

\item
Because of that change, strongly-seeing does not necessary imply
seeing.  The original rule for first-round voting required that for
$y$ to vote yes on $x$, $y$ must see $x$ and not merely be a
descendant of $x$.  But that meant that $y$ strongly seeing $x$ was
not enough for $y$ to vote yes for $x$, which breaks the proof of the
existence of famous witnesses (Theorem~\ref{thm:famous-exists}).

\item
Originally, when a peer had multiple famous witnesses, instead of
choosing one to consider unique, none of them were.  But then it was
not obvious that unique famous witnesses always exist.

\end{enumerate}

In addition to these changes, we added a careful proof of Progress
(Lemma~\ref{lem:progress}).  Also, a trivial difference is our first
round is 0, instead of 1.  This is preferable so every natural number
is a round number.

\subsection{Peers and Events}

To illustrate the abstract style we used to formalize the algorithm,
these are definitions of peers and events:

\begin{strictcode}
Parameter peer : Type.

Axiom peer_eq_dec
  : forall (a b : peer), {a = b} + {a <> b}.
endstrictcode

The axiom states that it is decidable whether two peers are the same.
This is a good example of
the sort of uncontroversial mathematical assumption we alluded to
in Section~\ref{sec:algorithm}.  A typical human-language proof would
not bother to make such an assumption explicit.

\begin{strictcode}
Parameter number_peers : nat.

Axiom number_of_peers
  : cardinality (@every peer) number_peers.

Axiom number_peers_minimum : number_peers >= 2.
endstrictcode

There exists a natural number that is the number of peers, and that
number is at least two.  (In the preceding we called it $N$.)  The
code ``\cd{@every peer}'' denotes the set of all peers.

\begin{strictcode}
Parameter honest : peer -> Prop.

Axiom supermajority_honest
  : supermajority honest every.
endstrictcode

The set of honest peers is a supermajority of the set of all peers.

\begin{strictcode}
Parameter event : Type.

Axiom event_eq_dec
  : forall (e f : event), {e = f} + {e <> f}.

Parameter creator : event -> peer.
endstrictcode

The equality of events is decidable, and every event has a creator.

\begin{strictcode}
Parameter parents
  : event -> event -> event -> Prop.

Axiom parents_fun :
  forall e e' f f' g,
    parents e f g
    -> parents e' f' g
    -> e = e' /\ f = f'.

Axiom parents_creator :
  forall e f g,
    parents e f g
    -> creator e = creator g.
endstrictcode

We write \cd{parents e f g} when {\tt e} and {\tt f} are the
self-parent and other-parent of {\tt g}.  The axioms say parents
are uniquely defined, and the self-parent has the same creator as the
event.

\begin{strictcode}
Definition initial (e : event) : Prop :=
  ~ exists f g, parents f g e.

Axiom initial_decide :
  forall e,
    initial e \/ exists f g, parents f g e.
endstrictcode

An initial event is one without parents, and it is decidable whether
an event is initial.

\begin{strictcode}
Inductive parent : event -> event -> Prop :=
| parent1 {e f g} :
    parents e f g
    -> parent e g

| parent2 {e f g} :
    parents e f g
    -> parent f g.

Definition self_parent (e f : event) : Prop :=
  exists g, parents e g f.

Axiom parent_well_founded
  : well_founded parent.
endstrictcode

We say \cd{parent x y} when {\tt x} is either parent of {\tt y}, and
\cd{self_parent x y} when {\tt x} is the self-parent of {\tt y}.  We
assume that the parent relation is well-founded, which means we can do
induction using that relation.  (In classical logic, that is equivalent
to saying there are no infinite descending chains.)  From that
assumption, we can show that the self-parent, strict ancestor, and
strict self-ancestor relations are also well-founded.

\subsection{Rounds}

In Baird~\cite{baird:hashgraph} the definitions of rounds and
witnesses are mutually dependent.  As here, a witness was the first
event created by a peer in a round.  Unlike here, to advance to the
next round an event would need to strongly see many {\em witnesses\/}
in the current round, while we require it only to strongly see many
{\em events} in the current round.  It is not hard to see that the definitions are
equivalent, but our version has the virtue that rounds can be defined
without reference to witnesses.

Our formalization of rounds then is:

\begin{strictcode}
Inductive round : nat -> event -> Prop :=
| round_initial {x} :
    initial x
    -> round 0 x

| round_advance {x y z m n A} :
    parents y z x
    -> round m y
    -> round n z
    -> supermajority A every
    -> (forall a,
          A a
          -> exists w,
               creator w = a
               /\ stsees w x  (* w << x *)
               /\ round (max m n) w)
    -> round (S (max m n)) x.
endstrictcodebeginstrictcode
| round_nadvance {x y z m n A} :
    parents y z x
    -> round m y
    -> round n z
    -> superminority A every
    -> (forall a w,
          A a
          -> creator w = a
          -> stsees w x
          -> exists i,
               i < max m n /\ round i w)
    -> round (max m n) x
endstrictcode

The first case says initial events belong to round 0.  In each of the
other two cases, $x$'s parents' rounds are $m$ and $n$.  Then, in the
second case, we advance to the next round ({\em i.e.,} give $x$ the
round $\max(m, n) + 1$) if there exists a set of peers $A$, comprising
a supermajority of the peers, such that for each peer $a \in A$, there
exists $w$ created by $a$ in round $\max(m,n)$ where $w \stsee x$.
The third case expresses the negation of that condition.  There
$S$ is a ``superminority'' of $T$ if $S$ contains at least one-third
of $T$'s elements.  That is, a superminority is a subset large enough to
deny a supermajority to its complement.

It is necessary to express the conditions for advancing or not
advancing to the next round as a pair of affirmative properties.  One
might imagine saying something like:

\begin{strictcode}
...
-> (b = true
    <->
    (.. x stsees enough from this round ..))
-> round (max m n + if b then 1 else 0) x
endstrictcode

\noindent
but with such a definition, the ``strongly sees enough'' portion could
not be implemented, since it could not mention rounds.  Any mention
of rounds in that position would be a non-positive occurrence of the
relation being defined, which is not permissible in an inductive
definition.

The benefit to going to the trouble to use an inductive definition
is one can do induction over the derivation of an event's round.  When
one is reasoning about rounds, this produces exactly the cases you
want: one base case, one inductive case where you advance, and one
where you don't.

\subsection{Voting}
\label{sec:formal-voting}

Expressing voting as an inductive definition is a bit tricky.  We
gave {\tt vote} the type:
\begin{code}
sample -> event -> event -> bool -> Prop
\end{code}
\noindent
Here, \cd{vote s x y b} means $y$'s vote on whether $x$ is
famous is $b$.  The $s$ is a point in the sample space, and can be
safely ignored for now.  (We discuss the sample space in
Section~\ref{sec:probability}.)

The case for the first round of voting is straightforward:

\begin{strictcode}
| vote_first {s x y m n v} :
    rwitness m x
    -> rwitness n y
    -> m < n
    -> m + first_regular >= n
    -> (Is_true v <-> x @ y)
    -> vote s x y v
endstrictcode

\noindent
Here, \cd{first_regular} is the formalization's name for the parameter
$d$, \cd{rwitness m x} means that $x$ is a round $m$ witness, and
\cd{@} means strict ancestor ($\sanc$, formalized as the transitive closure of 
{\tt parent}).  Then $y$ votes yes if $y$ is a strict descendant of $x$, and no
otherwise, provided $\round(x) < \round(y) \leq \round(x) +
d$.\footnote{This is a minor departure from the definition in
  Section~\ref{sec:voting}.  It
is convenient to allow witnesses to vote before the required $d$ rounds have
elapsed, but such early votes are ignored.}

The complications begin in the next case:

\begin{strictcode}
(* not a coin round *)
| vote_regular {s x y m n t f v} :
    rwitness m x
    -> rwitness n y
    -> m + first_regular < n
    -> (n - m) mod coin_freq <> 0
    -> election (vote s x) (pred n) y t f
    -> ((t >= f /\ v = true)
         \/ (f > t /\ v = false))
    -> vote s x y v
endstrictcode

\noindent
The first four premises say that more than $d$ rounds have elapsed
since $x$, and it is not currently a coin round.  Then
\cd{election (vote s x) (pred n) y t f}
says that when $y$ collects votes on $x$ from the previous round, it
receives $t$ yeas and $f$ nays.  Then $y$'s vote is true if $t \geq
f$ and false otherwise.

The trickiness lies in {\tt election}.  It has type:

\begin{code}
(event -> bool -> Prop)
  -> nat -> event -> nat -> nat -> Prop
\end{code}

\noindent
The first argument abstracts over the recursive call to {\tt vote}, in
order to disentangle {\tt election} and {\tt vote}.  We fill it in
with \cd{vote s x} to give it access to everyone's votes on $x$.  The
second argument is the round to collect votes from; we fill it in with
$\round(y) - 1$.

An auxiliary definition, \cd{elector n w y}, specifies the round $n$
witnesses $w$ that can send votes to $y$:

\begin{strictcode}
Definition elector (n : nat) (w y : event) :=
  rwitness n w /\ stsees w y.
endstrictcode

\noindent
At
this point, we might imagine defining \cd{election V n y t f} as
follows:

\begin{strictcode}
cardinality
  (fun w =>
     elector n w y /\ V w true) t
/\
cardinality
  (fun w =>
     elector n w y /\ V w false) f
endstrictcode

\noindent
Indeed, this is a perfectly good definition.  However, if we define
{\tt election} this way, {\tt vote} is not allowed to call it as
above.  Since {\tt vote} calls {\tt election} with a recursive
instance of {\tt vote}, {\tt election} must use $V$ only positively.
This definition does not, because cardinality does not use its
first argument only positively.  Thus {\tt vote} would be ill-defined.

Instead, we code around the problem:

\begin{strictcode}
Definition election 
  (V : event -> bool -> Prop) 
  (n : nat) (y : event) (t f : nat)
  :=
  exists T F,
    cardinality
      (fun w => elector n w y) (t + f)
    /\ cardinality T t
    /\ cardinality F f
    /\ (forall w, T w
          -> elector n w y /\ V w true)
    /\ (forall w, F w 
          -> elector n w y /\ V w false)
endstrictcode

\noindent
We existentially quantify over two sets $T$ and $F$, where
the intention is that $T$ is the yeas and $F$ the nays.  The
final two conjuncts check that every element of $T$ ($F$) is indeed an
elector and votes yes (no).  The previous two conjuncts check that the
cardinalities of $T$ and $F$ are $t$ and~$f$.

Finally, the first conjunct ensures that we are not missing any votes:
If we assume that $V$ does not allow a single witness to vote both
ways (which \cd{vote s x} will not), it follows that $T$ and $F$ are
disjoint.  Thus, in order to check that every vote is accounted for,
we need only check that there are $t + f$ electors total.

\smallskip

The remaining two cases of {\tt vote} implement coin rounds.  They
call {\tt election} in the same manner as above, in order to find out
whether a supermajority already exists.

\begin{strictcode}
(* coin round but supermajority exists *)
| vote_coin_super {s x y m n t f} {v : bool} :
    rwitness m x
    -> rwitness n y
    -> m < n
    -> (n - m) mod coin_freq = 0
    -> election (vote s x) (pred n) y t f
    -> (if v then t else f)
         > two_thirds number_peers
    -> vote s x y v
endstrictcode
\begin{strictcode}
(* coin round and no supermajority exists *)
| vote_coin {s x y m n t f} :
    rwitness m x
    -> rwitness n y
    -> m < n
    -> (n - m) mod coin_freq = 0
    -> election (vote s x) (pred n) y t f
    -> t <= two_thirds number_peers
    -> f <= two_thirds number_peers
    -> vote s x y (coin y s)
endstrictcode

\noindent
In the final case, \cd{coin y s} gives $y$'s pseudo-random coin flip.

\paragraph{Induction}

Although Rocq accepts the definition of {\tt vote}, since all its
recursive occurrences are positive, the definition is still too
complicated for Rocq to give it a useful induction principle.  It does,
however, provide a useful case-analysis principle.  Thus, Rocq
essentially promises that {\tt vote} is a fixed point, but not a
{\em least\/} fixed point.

This is inconvenient, because there are many times one would like to
do induction over votes.  Fortunately, we can work around the
problem.  The recursive instances of {\tt vote} always deal with
strict ancestors of the vote in question, so one can employ well-founded
induction using the strict-ancestor relation.  Within that induction, one can
then do a case analysis over the vote in question.  This provides the
power of the induction principle that one would have liked Rocq to
provide automatically.

\subsection{Worlds}

We use {\em worlds\/} to talk about potentially incompatible
evolutions of the hashgraph.  (The term is motivated by the use of the
term in Kripke models.)  In a modest loosening of the rules from our
informal presentation, we allow that an event $x$ might have two
distinct self-children $y$ and $z$, even if $x$'s creator is honest,
provided that $y$ and $z$ exist only in distinct futures.  However, $y$
and $z$ can never coexist in the same future (again, if the creator is
honest).  That is, two events {\em in the same world\/} cannot form a
fork on an honest peer.

Thus, a world is a set of events that is closed under ancestor, and
contains no forks:\footnote{Note that nothing requires a world to
contain {\em all\/} the
events, even at a particular point in time.  Thus, for example, we can
also use worlds to express the knowledge that a particular peer has
about the hashgraph.}

\begin{strictcode}
Definition fork (e f : event) : Prop :=
  creator e = creator f
  /\ ~ e $= f
  /\ ~ f $= e.

Record world : Type :=
mk_world
  { member : event -> Prop;

    world_closed :
      forall x y,
        x @= y -> member y -> member x;

    world_forks :
      forall a,
        honest a
        -> ~ exists e f,
               member e /\ member f /\
               fork e f /\ creator e = a }.
endstrictcode

\noindent
Here \cd{@=} means ancestor ($\anc$), and \cd{$=} means self-ancestor
($\selfanc$).

Many facts rely only on the events themselves, not any worlds they
inhabit, but some require a world assumption.  For
example, consider the strongly seeing lemma (Lemma~\ref{lem:strongly-seeing}):

\begin{strictcode}
Lemma strongly_seeing :
  forall W x y v w,
    member W v
    -> member W w
    -> fork x y
    -> stsees x v
    -> stsees y w
    -> False.
endstrictcode

\noindent
Recall that, in the proof, we obtained events $z$ and $z'$, sharing the same honest
creator, where $x \see z \anc v$ and $y \see z' \anc w$.  We concluded
that either $z \selfanc z'$ or $z' \selfanc z$, since otherwise $z$
and $z'$ would constitute a fork on a honest peer.  But in the
formalization, forks are possible---even on an honest peer---if the
sides of the fork belong to different futures.  Requiring $v$ and
$w$ (and therefore $z$ and $z'$) to belong to the same world rules that out.

\subsection{Probability}
\label{sec:probability}

For the probabilistic elements of the proof, we employ a lightweight
probability library.  The library is lightweight in three senses:

\begin{enumerate}
\item
The axioms of real arithmetic, measure theory, and probability are
given as Rocq axioms.  We did not build a model of those notions in Rocq.

\item
The library uses only basic Rocq: just definitions, and theorems using
Ltac.  Since we do not use Rocq's more advanced features, the library
is easy to use and is robust to changes in Rocq.

\item
The library contains only what was necessary for its application to
this work.  For example, the library contains no treatment of
continuous probability.
\end{enumerate}

The treatment of real arithmetic makes some effort to respect
constructive niceties.\footnote{One axiom ($x \neq y$ implies $x < y
  \mathrel{\vee} y < x$) requires Markov's principle, which is
  quasi-constructive in the sense that it does not impede the
  extraction of computational content.  The rest are
  intuitionistically valid.  The (resolutely classical)
  least-upper-bound axiom ({\em i.e.,} every bounded, nonempty set has
  a least upper bound) takes the law of the excluded middle as a
  premise, so it is intuitionistically valid, vacuously.}  However,
the treatment of probability theory is fully classical, as we did not
intend to explore the issues of constructive probability in this work.
Any proof of a probability measure may employ the law of the excluded
middle.  Nevertheless, since the probability measure does not provide
computational content anyway, one might argue that little is sacrificed
by permitting classical reasoning while establishing probability
measures.

\medskip

As usual, the central elements of probability are a sample space and a
probability predicate.  A point in the sample space determines the
outcome of every random or nondeterministic action, as well as the
behavior of the adversary.  The probability predicate relates subsets
of the sample space to real numbers.

\begin{code}
sample : Type
number : Type
prob : (sample -> Prop) -> number -> Prop
\end{code}

A probability event is a subset of the sample space.  A probability
event might be unmeasurable, but the probability of measurable events
is unique.

Since a point in the sample space determines everything that takes
place, and specifically determines what events get created, a sample
determines a world.  We call this the {\em global\/} world determined
by that sample.

\begin{code}
global : sample -> world
\end{code}

\noindent
For now we can think of {\tt global} as primitive, but once we develop
some more machinery, we can actually define it.

At this point we can state the good-coins predicate along the lines we
discussed in Section~\ref{sec:termination}:

\begin{bigstrictcode}
Definition good_coins (s : sample) : Prop :=
  forall m x,
    rwitness m x
    -> exists n y t f v,
         m < n
         /\ (n - m) mod coin_freq = 0
         /\ first_witness s n y
         /\ election (vote s x) (pred n) y t f
         /\ ((t >= f /\ v = true) 
             \/ (f > t /\ v = false))
         /\ forall w,
              member (global s) w
              -> rwitness n w
              -> honest (creator w)
              -> (forall z,
                    member (global s) z
                    -> rwitness (2 + n) z
                    -> realtime s w z)
              -> coin w s = v.
endbigstrictcode

The proposition \cd{realtime s w z}\, says that (given sample point
$s$), $w$ was created (non-strictly) before $z$.  (It is defined using
spawn order, which we discuss Section~\ref{sec:eventual-agreement}.)  The
proposition \cd{first_witness s n y}\, says that (given sample point
$s$) $y$ is the first (in the sense of {\tt realtime}) witness in
round $n$.

The eventual agreement theorem says that {\tt good\_coins} always
holds:

\begin{bigcode}
Theorem eventual_agreement : prob good_coins one.
\end{bigcode}

Eventual agreement establishes the first condition of the termination
theorem, but we state termination in absolute rather than
probabilistic terms:

\begin{strictcode}
Theorem fame_consensus :
  forall s x,
    good_coins s
    -> member (global s) x
    -> witness x
    -> exists y v,
         member (global s) y
         /\ decision s x y v.
endstrictcode

From the two theorems it follows\footnote{if we assume the probability
  measure is complete} that we reach consensus with probability one,
but it's preferable to state matters in logical rather than
probabilistic terms when possible because (1)~the result is slightly
stronger, and (2) logic is much easier to work with than probability.

\subsubsection{Eventual agreement}
\label{sec:eventual-agreement}

The heart of the eventual agreement proof is showing that the fair
coins are independent of each other (easy) and of the target result
{\tt v} (hard).  This in turn relies on an independence axiom:

\begin{code}
Axiom coin_independent :
  forall x P Q,
    honest (creator x)
    -> settledby P x Q
    -> measurable P
    -> measurable Q
    -> cindep (fun s => coin x s = true) P Q.
\end{code}

\noindent
Suppose $P$ and $Q$ are measurable probability events and $x$ is an
honest event.  Then the axiom says that $x$'s coin is conditionally
independent of $P$ (assuming $Q$), provided that $P$ is settled before
$x$ is created (assuming $Q$).

Using this axiom, we can reason about independence without knowing how
to compute the probability of $P$ and $Q$, which would require knowing
how the network and the adversary operate.  Instead, we merely assume
that neither can use information from the future.

\smallskip

We say that $P$ is settled before $x$ (assuming $Q$) if $P$ cannot
distinguish between sample points that satisfy $Q$ and are similar
until $x$:

\begin{code}
Definition settledby 
  (P : sample -> Prop) (x : event) 
  (Q : sample -> Prop) :=
    forall s s',
      Q s
      -> Q s'
      -> similar x s s'
      -> P s <-> P s'.
\end{code}

To define ``similar'' we need some auxiliary notions.  First, we
need a notion of time.  We get that from {\em spawn order}:

\begin{code}
spawn : sample -> nat -> event
\end{code}

\noindent
Suppose $s$ is a point in the sample space.  Then there exists a set
$S$ of all the events that will be created in the history $s$
determines.  (In order words, $S$ contains all the elements of
\cd{global s}.)  Sort $S$ according to the order the events are
created in real time to obtain $x_0, x_1, x_2, \ldots$.  Then we define
$\mbox{\cd{spawn s i}} = x_i$.  Obviously, the ``real time'' ordering
cannot be used by any participant; it is used only in reasoning about
the algorithm.

Several axioms 
(appearing in Figure~\ref{fig:spawn-axioms})
govern spawn order: (1) No event spawns more than once.
(2) An event's parents spawn
before it does.  
The next two
axioms are general rules of hashgraphs that are
convenient to formalize in terms of spawn order: 
(3) Honest peers do not create forks.  (4) Every honest peer
eventually communicates with every other honest peer.  (5) Finally, it is
convenient to exclude any event that is part of no future.

\begin{figure}[t]
\fighead
\begin{strictcode}
Axiom spawn_inj :
  forall s i j,
    spawn s i = spawn s j -> i = j.

Axiom spawn_parent :
  forall s x i,
    parent x (spawn s i)
    -> exists j, x = spawn s j /\ j < i.

Axiom spawn_forks :
  forall s i j,
    honest (creator (spawn s i))
    -> fork (spawn s i) (spawn s j)
    -> False.

Axiom honest_peers_sync :
  forall s i a b,
    honest a
    -> honest b
    -> a <> b
    -> exists j k,
         i <= j
         /\ j <= k
         /\ creator (spawn s j) = a
         /\ creator (spawn s k) = b
         /\ spawn s j @ spawn s k.

Axiom no_orphans :
  forall x, exists s i, x = spawn s i.
endstrictcode
\caption{Spawn order axioms}
\label{fig:spawn-axioms}
\figfoot
\end{figure}

We can define the {\tt realtime} order using spawn order in the
obvious way.
We can also define the global world \cd{global s} as the
set of all events $x$ such that \cd{exists i, spawn s i = x}.  The
\cd{spawn_parent} axiom provides the \cd{world_closed} specification,
and \cd{spawn_forks} provides \cd{world_forks}.

\smallskip

Now we say that $s$ and $s'$ are similar until $x$ if (1) $x$ spawns
at the same time in each (or does not spawn at all), and (2) until then,
both agree on the events that are spawned and their coins.  If
$x$ never spawns, then $s$ and $s'$ must agree forever. 

\begin{bigstrictcode}
Definition similar (x : event) (s s' : sample) :=
  (forall i, x = spawn s i <-> x = spawn s' i)
  /\
  (forall i,
     (forall j, x = spawn s j -> i < j)
     -> spawn s i = spawn s' i
        /\ coin (spawn s i) s 
           = coin (spawn s i) s').
endbigstrictcode

\noindent
As a boundary case, observe that $x$'s coin is {\em not\/} settled
before $x$.  Also note that other matters might differ between $s$ and
$s'$, but the consensus algorithm depends only on events and their
coins.

\section{Fairness}
\label{sec:fairness}

We also formalized an unpublished fairness theorem by
Baird~\cite{baird:personal-commo}.

\begin{lemma}
\label{lem:pre-fairness}
Suppose $d \geq 2$.  Suppose also that $x$ is a round $i$ witness.  If
there exists $y$ such that $x \stsee y$ and both parents of $y$
belong to a round no later than $i$, then $x$ will be famous.
\end{lemma}
\begin{proof}
Suppose $w$ is an arbitrary round $i + d$ witness.  We claim that $x
\anc w$.  It follows that $w$ votes yes.  Since $w$ is arbitrary,
every round $i + d$ witness votes yes, so $x$ will be decided to be
famous in round $i + d + 1$.

Since $x \stsee y$, the set $U = \{ u \sbar x \see u \anc y \}$ is
supermajor.  By Lemma~\ref{lem:round-elim}, there exists a supermajor
set $V$ of round $i + d - 1$ witnesses such that forall $v \in V$, $v
\stsee w$.  The intersection of three supermajorities is nonempty, so
there exists an honest peer $a$ and events $u$ and $v$ created by $a$
such that $x \see u \anc y$ and $v \stsee w$.  Since $a$ is honest,
either $u \selfanc v$ or $v \selfsanc u$.

If $u \selfanc v$ then we are done, since $x \anc u \anc v \anc w$, so
let us assume $v \selfsanc u$.  Then $v \sanc y$.  Thus $v$ is an
ancestor of one of $y$'s parents, both of whom belong to a round no
later than $i$.
But $v$ belongs to round $i + d - 1$, which is at least $i + 1$ since
$d \geq 2$.  This is a contradiction, since $v$ cannot belong to a
later round than any of its descendants.
\end{proof}

\begin{theorem}[Fairness]
If $d \geq 2$ then every round's set of famous witnesses is supermajor.
\end{theorem}
\begin{proof}
Let $i$ be arbitrary.  By Lemma~\ref{lem:progress}, there exists
a round $i+1$ witness $x$.  By following back $x$'s ancestry, we can
obtain a round $i+1$ witness $y$ (possibly $x$ itself), both of whose
parents belong to round at most $i$.  Since neither of $y$'s parents belong to
round $i+1$, $y$ must strongly see a supermajor set of round $i$
witnesses.  By Lemma~\ref{lem:pre-fairness}, all those witnesses will
be famous.
\end{proof}

The significance of the fairness theorem is that it means a majority of
the famous witnesses in any round must be honest.  Consequently, every
event's consensus timestamp is governed by honest peers.  The
timestamp might come from a dishonest peer, but even if so, it will be
bracketed on both sides by timestamps from honest peers.

The formalization of the fairness theorem and its consequences
regarding timestamps is just under 700 lines of Rocq.

\section{Additional Related Work}

Much recent work on Byzantine fault tolerance has set aside the
asynchronous model of the network in favor of partially synchronous
models~\cite{dwork+:consensus-partial-synchrony}. In partially
synchronous models, the time it takes to deliver a message is bounded,
but that bound cannot be fully exploited for some reason.  (For
example, the bound exists, but is not known.)  Partial synchrony
can allow for simpler consensus algorithms, such as
HotStuff~\cite{yin+:hotstuff}, PaLa~\cite{chan+:pala},
Streamlet~\cite{chan+:streamlet}, and Simplex~\cite{chan+:simplex}.

There has been some prior work on machine-checked verification of
Byzantine consensus protocols.
Kellom\"{a}ki~\cite{kellomaki:paxos-superposition-pvs} verified a
version of the Paxos protocol~\cite{lamport:paxos} using
PVS~\cite{owre+:pvs}.
Lamport~\cite{lamport:byzantizing-paxos-by-refinement} partially
verified the BFT
protocol~\cite{castro+:practical-byzantine-fault-tolerance} using
TLA+~\cite{cousineau+:tlaps}.

Both of those, like this work, verified theoretical presentations of
the protocol.  Hawblitzel, {\em et al.}~\cite{hawblitzel+:ironfleet}
verified the code for a practical implementation of a Paxos-like
algorithm written in Dafny~\cite{leino:dafny}.  To do something
similar for Hashgraph is the ultimate aim of this work.

\section{Further Developments and Future Work}
\label{sec:future-work}

The ultimate aim of this work is a fully verified implementation of
the Hashgraph consensus algorithm.  There are a number of differences between the algorithm
given here and the one that is implemented.  We have incorporated many
of them into the formalization already, but some others are future work.

\begin{itemize}
\item
We have completed a version of the algorithm that supports weighted peers
(a.k.a., proof-of-stake) instead of giving an equal weight to every
peer.  In that version, a supermajority of the weight must belong to
honest peers.  This change is largely straightforward, but there were
some complications in the re-proof of Theorem~\ref{thm:famous-exists}
stemming mainly from the need for a weighted version of the pigeonhole
principle.

\item
The algorithm here is a batch algorithm, while the implemented system
is online.  Moving to an online version requires two main
changes:

\begin{itemize}
\item
We cannot permanently blacklist peers who create a fork, since this
would require retaining information about them indefinitely.  Instead
we introduce a consistent way of establishing priority between the
forked events, and one can only ``see'' the higher priority event.
This is sufficient to reestablish Lemma~\ref{lem:strongly-seeing},
since the heart of the proof was the construction of an impossible
event that sees both sides of a fork, and that remains impossible.
This development is complete.

\item
Since the participants of the network change over time, we need a way
to deal with changing weights.  (When a participant leaves the
network, we can view that as its weight going to zero.)  This is
future work, although a prerequisite (having weights at all) is
already done.  The main complication is ensuring that all peers
always agree on all the weights, when the weights are determined by
previous transactions.
\end{itemize}

\item
The algorithm includes some operations that are expensive to perform
and that we want to avoid as much as possible.  A good example is
strongly seeing, which requires one to count all the different peers
one can pass through between one event and another.  That involves
exploring many different paths.  However, one can show that 
one can limit oneself to exploring certain canonical paths without
sacrificing any key properties.  A version incorporating this and
other optimizations is complete.
\end{itemize}

In addition, one would like to establish the fairness theorem for $d=1$.  This is an
important area for research, since $d=1$ is preferable (faster
consensus) but provably honest timestamps are also desirable.

\bibliographystyle{plain}
\bibliography{c:/crary/crary}

\end{document}

%% file: defs.tex
\mathchardef\col="003A  
\mathchardef\semi="603B 
\mathchardef\bang="6021 

\mathchardef\lt="313C  
\mathchardef\gt="313E  

\newdimen\cascadeindent
\newdimen\cascdimen
\newcommand{\cascitem}{\\\global\advance\cascdimen by\cascadeindent\hspace{\cascdimen}}

\newcommand{\cascback}[1]{\\\global\advance\cascdimen by-#1.0\cascadeindent\hspace{\cascdimen}}

\newcommand{\bigleft}{\left(\begin{array}{@{}l@{}}}
\newcommand{\bigright}{\end{array}\right)}

\newdimen\proofadjustabove
\newdimen\proofadjustbelow
\newenvironment{proof}{\vspace{\proofadjustabove}\noindent\rm{\bf Proof}\vspace{\proofadjustbelow}\begin{list}{}{\leftmargin=1.5em\labelsep=1.5em\labelwidth=0pt}\item}{\end{list}}
\newenvironment{proofsketch}{\vspace{\proofadjustabove}\noindent\rm{\bf Proof Sketch}\vspace{\proofadjustbelow}\begin{list}{}{\leftmargin=1.5em\labelsep=1.5em\labelwidth=0pt}\item}{\end{list}}

\newtheorem{theorem}{Theorem}[section]
\newtheorem{lemma}[theorem]{Lemma}
\newtheorem{corollary}[theorem]{Corollary}

\newtheorem{definition}[theorem]{Definition}
\newenvironment{defn}{\begin{definition}\rm}{\end{definition}}

\chardef\at="40
\chardef\tilde="7E
\chardef\caret="5E
\chardef\lb="7B
\chardef\rb="7D

\newcommand{\figfoot}{\vspace{1ex}\hrule}
\newcommand{\fighead}{\hrule\vspace{1.5ex}}
\newcommand{\zilch}{\mbox{}}

\newcommand{\anc}{\leq}
\newcommand{\modulo}{\mathbin{\it mod}}
\newcommand{\sanc}{\lt}
\newcommand{\sbar}{\:|\:}
\newcommand{\see}{\trianglelefteq}
\newcommand{\selfanc}{\sqsubseteq}
\newcommand{\selfsanc}{\sqsubset}
\newcommand{\stsee}{\ll}
\newcommand{\twothirds}{\nicefrac{2}{3}}

\newcommand{\backward}{\mbox{\sf backward}}
\newcommand{\coin}{{\sf coin}}
\newcommand{\creator}{{\sf creator}}
\newcommand{\decide}{{\sf decide}}
\newcommand{\election}{{\sf election}}
\newcommand{\no}{{\sf no}}
\newcommand{\round}{{\sf round}}
\newcommand{\vote}{{\sf vote}}
\newcommand{\yes}{{\sf yes}}